\documentclass[a4paper,10pt,openright]{article} % Per avere margini destro e sinistro diversi
\usepackage[english]{babel} % per lingua. TeXnic dà warning per misteriosi motivi...
\usepackage[italian]{}
\usepackage[latin1]{inputenc} % per gli accenti e carattere
\usepackage{amsmath,amssymb, stmaryrd,amsthm,amsfonts,amscd,color,enumerate,eucal,latexsym,mathrsfs,mathtools,cancel,tikz}
\usepackage{epsfig, hieroglf, protosem}  % ciao
\usepackage{simplewick, bm}
\usepackage{hyperref}
\usepackage[all,cmtip]{xy}
\usetikzlibrary{matrix,calc}
\usepackage{verbatim}
\numberwithin{equation}{section}

\newtheoremstyle{TheoremStyle}% <name>
{3pt}% <Space above>
{3pt}% <Space below>
{\slshape}% <Body font>
{}% <Indent amount>
{\bf}%{\itshape}% <Theorem head font>
{:}% <Punctuation after theorem head>
{.5em}% <Space after theorem head>
{}% <Theorem head spec (can be left empty, meaning 'normal')>

\theoremstyle{TheoremStyle}
\newtheorem{theorem}{Theorem}

\newtheorem{proposition}[theorem]{Proposition}
\newtheorem{lemma}[theorem]{Lemma}
\newtheorem{definition}[theorem]{Definition}
\newtheorem{remark}[theorem]{Remark}
\newtheorem{example}[theorem]{Example} %[theorem]{Example}

\title{Reconstruction Theorem for Germs of Distributions on Smooth Manifolds}

\author{
	Paolo Rinaldi \thanks{PR: Dipartimento di Fisica,
	Universit\`a degli Studi di Pavia \& INFN, Sezione di Pavia, 
	Via Bassi 6,
	I-27100 Pavia,
	Italia; Istituto Nazionale di Alta Matematica, Sezione di Pavia, via Ferrata 5, 27100 Pavia, Italia
	\mbox{paolo.rinaldi01@universitadipavia.it}
}
\and
	Federico Sclavi \thanks{FS: Dipartimento di Fisica,
	Universit\`a degli Studi di Pavia \& INFN, Sezione di Pavia, 
	Via Bassi 6,
	I-27100 Pavia,
	Italia; Istituto Nazionale di Alta Matematica, Sezione di Pavia, via Ferrata 5, 27100 Pavia, Italia
\mbox{federico.sclavi01@universitadipavia.it}}
}

\begin{document}
\maketitle
\begin{abstract}
The \emph{reconstruction theorem} is a cornerstone of the theory of regularity structures \cite{Hai14}. In \cite{CZ20} the authors formulate and prove this result in the language of distributions theory on the Euclidean space $\mathbb{R}^d$, without any reference to the original framework. In this paper we generalize their constructions to the case of distributions over a generic $d$-dimensional smooth manifold $M$, proving the reconstruction theorem in this setting. This is done having in mind the extension of the theory of regularity structures to smooth manifolds.
\end{abstract}
%\tableofcontents

\section{Introduction}
The \emph{Reconstruction Theorem} is one of the cornerstones of the Theory of Regularity Structures \cite{Hai14}, the framework in which this theorem was first formulated.
This theory provides a milestone in the analysis of stochastic partial differential equations on the Euclidean space $\mathbb{R}^d$, which is the original motivation for this theory, since it allows to apply fixed point techniques to such equations. \\
Stochastic partial differential equations are also closely related to quantum field theory, in particular through stochastic quantization \cite{PW81}, which links stochastic PDEs with the path integral formulation of Euclidean field theory. The idea at the heart of stochastic quantization is construct the path integral measure of an Euclidean interacting field theory as the invariant measure of a stochastic process whose dynamics is ruled by a parabolic non-linear stochastic PDE. 
Recently, again in the interplay between stochastic PDEs and quantum field theory, there are also some efforts to apply the techniques proper of the latter to problems of the former, in particular for renormalization \cite{DDRZ20}.

Nonetheless, on account of general relativity, a natural and more general setting for quantum field theory is represented by curved spacetimes.
As a consequence, from the point of view of quantum field theory on curved spacetimes \cite{BFDY15, book, Fredenhagen:2014lda, DDR20}, in order to extend this fruitful interaction with stochastic PDEs also to this framework, it would be desirable to have a formulation of the theory of regularity structures on a smooth manifold and a first step in this direction should be the formulation of the reconstruction theorem on a smooth manifold. This is the aim of the present paper.
There are already some efforts to this end \cite{DDK19}, where the authors consider the Riemannian case.

Recently, in \cite{CZ20}, the authors proved that this result can be formulated as a result of distributions theory on the Euclidean space $\mathbb{R}^d$, \emph{i.e.}, in $\mathcal{D}^\prime(\mathbb{R}^d)$, without any reference to the theory of regularity structures.
In particular, the problem is formulated in the following way.
If for any $x\in\mathbb{R}^d$ we are given a distribution $F_x\in\mathcal{D}^\prime(\mathbb{R}^d)$, one may wonder whether there exists a distribution $f\in\mathcal{D}^\prime(\mathbb{R}^d)$ which is \emph{locally} approximated, in a suitable sense, by $F_x$ in a neighbourhood of $x\in\mathbb{R}^d$, for any $x\in\mathbb{R}^d$.
In their paper, the authors proved that this is actually the case under a further hypothesis, dubbed \emph{coherence}, providing a bound for the difference $F_x-F_y$ for $y$ and $x$ sufficiently close. 
This condition, which is closely related to the generalized H\"older condition, is inspired, in the language of regularity structures, by the notion of \emph{model} and of \emph{modelled distribution}.

From the viewpoint of the extension to the manifold setting, the main advantage of the formulation of \cite{CZ20} is that it is formulated in a purely distributional language.
Since distributions have an intrinsic \emph{local} nature, they can be considered in a very natural way also on generic smooth manifold $M$ \cite{Hor03, BGP07}. This argument makes this version of the reconstruction theorem the most convenient for the extension to the smooth manifold setting. 

To this end we translate the notion of germ of distributions and the key notion of coherence to the case of a smooth manifold, yielding a more local definition of these notions. 
Nonetheless, in Proposition \ref{Prop: independence from the atlas} we prove that the notion of coherence we give is actually independent of the atlas, in agreement with the locality of the whole construction. As a consequence we have a geometric notion of coherence.
\\

The main result of this paper is the reconstruction theorem, see Theorem \ref{Thm: reconstruction}, for $\gamma$-coherent germs of distributions, with $\gamma>0$, see Section \ref{Sec: main definitions} for details. 
In particular, with the notation of Definition \ref{Def: coherence}, the main result is the following.
\begin{theorem}\label{Thm: main result}
Let $M$ be a $d$-dimensional smooth manifold and let $\mathcal{A}=\{(U_j,\phi_j)\}_{j}$ be an atlas over $M$. 
Let $\gamma>0$ and let $F=(F_p)_{p\in M}$ be a $\gamma$-coherent germ of distributions on $(M,\mathcal{A})$. 
There exists a unique distribution $\mathcal{R}F\in\mathcal{D}^\prime(M)$ such that, for any local chart $(U,\phi)\in\mathcal{A}$, $\phi_*(\mathcal{R}F)\in\mathcal{D}^\prime(\phi(U))$ satisfies, for any compact set $K\subset U$ and for any $h\in\mathcal{D}(\phi(U))$,
\begin{align*}
|(\phi_*(\mathcal{R}F)-\phi_*(F_p))(h_{\phi(p)}^\lambda)|\lesssim\lambda^\gamma\,,
\end{align*}
uniformly for $p\in K$ and for $\lambda\in(0,1]$.
\end{theorem}

This result is proven as a consequence of a localized (on an open set) version of the reconstruction theorem of \cite{CZ20} and of the very characterization of the notion of distribution on a smooth manifold -- see Appendix \ref{Appendix: distribution on manifold} and \cite{Hor03}.
A further advantage of our result is that it shows that the reconstruction theorem holds true already at the level of smooth manifolds, without calling for further structures, such as Riemannian ones.

Finally, we discuss in detail the dependence of the reconstruction on the the atlas for $\gamma$-coherent germs of distributions with $\gamma>0$.
In particular, in Theorem \ref{Thm: coherence wrt different atlases} we prove that, in such a scenario, the reconstruction is independent of the atlas.

In the Euclidean space $\mathbb{R}^d$ setting, if the coherence parameter $\gamma$ is non-positive, one has existence of the reconstruction, yet without uniqueness.
This result can be achieved also on the smooth manifold setting, as we discuss in Theorem \ref{Rmk: gamma<0}, where we prove existence without uniqueness of the global reconstructed distributions. We underline that, in addition to being non-unique, these global reconstructed distributions depend on the atlas and on the partition of unity used to construct them.

\paragraph{Outline of the Paper}
The paper is organized as follows: in Section \ref{Sec: main definitions} we introduce the notion of \emph{germ of distributions} on a smooth manifold $M$ and the notion of \emph{coherence}, which is the key to the reconstruction. 
In this section we also discuss \emph{enhanced coherence}. Moreover, in this section we prove that coherence does not depend on the atlas.\\
In Section \ref{Sec: Reconstruction theorem} we state and prove the \emph{reconstruction theorem} for a $\gamma$-coherent germ of distributions on a smooth manifold, with $\gamma>0$. 
In the same section we also discuss the independence of the reconstruction from the atlas for $\gamma>0$. Eventually, we state and prove the reconstruction theorem for $\gamma$-coherent germs of distributions with $\gamma\leq0$. \\
Finally, on the one hand, in Appendix \ref{Appendix: distribution on manifold} we shall recall some notions of distributions theory on smooth manifolds in order for the paper to be self-contained. On the other hand, in Appendix \ref{Appendix: enhanced coherence on an open set} we discuss coherence and enhanced coherence on an open set of the Euclidean space $\mathbb{R}^d$, which is a propedeutical case study to the case of a smooth manifold.
 
\paragraph{Notation}
In the following, $M$ will denote a $d$-dimensional connected smooth manifold such that $\partial M=\emptyset$. Moreover, we will endow the manifold $M$ with the Borel $\sigma$-algebra. The pair $(U,\phi)$ denotes a generic local chart of $U$: \emph{i.e.}, $U\subset M$ is an open set and $\phi:U\to\phi(U)\subset\mathbb{R}^d$ is a diffeomorphism, representing a coordinate on $U$. Given a generic function $f:M\to N$, with $M$ and $N$ smooth manifold, $f^*$ and $f_*$ shall denote the pull-back and the push-forward, respectively, via this map.

We write $\mathcal{D}(M)$ for the space of smooth and compactly supported functions over $M$, endowed with the usual locally convex topology and $\mathcal{D}^\prime(M)$ shall denote the space of distributions over $M$, see Appendix \ref{Appendix: distribution on manifold} for further details. Moreover, $B(0,1)\subset\mathbb{R}^d$ will be the unitary ball centred at the origin. 
Given $U\subset\mathbb{R}^d$ and a function $f\in\mathcal{D}(U)$, we introduce the following rescaled version of this function, for $x\in U$, 
\begin{align}\label{Eq: scaled test-function}
f_x^\lambda(y)\vcentcolon=\lambda^{-d}f(\lambda^{-1}(y-x))
\end{align}
for $\lambda\in(0,1]$. We shall also adopt the following convention: $f^\lambda\equiv f^\lambda_0$.
In the following, we shall integrate test-functions $f\in\mathcal{D}(\mathbb{R}^d)$ with respect to the Lebesgue measure $dx$ on $\mathbb{R}^d$. This is just for convenience, a priori we could consider any measure on $\mathbb{R}^d$ which is absolutely continuous with respect to the Lebesgue measure.
Eventually, the symbol $\lesssim$ shall denote the inequality up to a multiplicative finite constant. 

\paragraph{Acknowledgements.}
We are thankful to  F. Caravenna, C. Dappiaggi, N. Drago and L. Zambotti for helpful discussions on the topic. We are thankful to  C. Dappiaggi for the valuable comments on a first version of this manuscript. The work of the authors is supported by a Doctoral Fellowship of the University of Pavia.
The work of P.R. was partly supported by a fellowship of the ``Progetto Giovani GNFM 2019" under the project ``Factorization Algebras vs AQFTs on Riemannian manifolds" fostered by the National Group of Mathematical Physics (GNFM-INdAM).

\section{Main Definitions}\label{Sec: main definitions}
In this section we shall define the main tools we are going to use in the paper. We shall also prove some of their properties.

Following \cite{CZ20}, we start by introducing the notion of \emph{germ} of distributions. 
\begin{definition}
Let $M$ be a $d$-dimensional smooth manifold. We define \emph{germ of distributions} over $M$ a family $F=\{F_p\}_{p\in M}$ of distributions, $F_p\in\mathcal{D}^\prime(M)$ for any $p\in M$, such that, for any $h\in\mathcal{D}(M)$, $p\mapsto F_p(h)$ is a measurable map with respect to the Borel $\sigma$-algebra of the manifold $M$.
\end{definition}

\begin{remark}
The idea at the heart of the notion of germ of distributions is that, under a further assumption which is \emph{coherence}, $F_p$ can be seen as an approximation, locally at the point $p\in M$, of a global distribution $\mathcal{R}F$. The idea of the reconstruction theorem is that of associating, under suitable assumptions, this (unique) global distribution with the germ of distributions.
\end{remark}
We shall now define the notion of coherent germ of distributions on a manifold, which is the key for the reconstruction theorem.
\begin{definition}\label{Def: coherence}
Let $M$ be a smooth $d$-dimensional manifold and let $\mathcal{A}=\{(U_\alpha,\phi_\alpha)\}_{\alpha}$ be a smooth atlas on $M$.
Let $F=(F_p)_{p\in M}$ be a germ of distributions on $M$ and $\gamma\in\mathbb{R}$. We say that $F$ is $\gamma$-coherent on $(M,\mathcal{A})$ if for any $(U,\phi)\in\mathcal{A}$ there exists $f\in\mathcal{D}(\phi(U))$, with $\int_{\mathbb{R}^d} dx\,f(x)\neq0$, such that for any compact set $K\subset U$ there exists $\alpha_K^U\leq\min\{0,\gamma\}$ such that
\begin{align}\label{Eq: coherence condition}
\vert(\phi_*(F_p)-\phi_*(F_q))(f^\lambda_{\phi(q)})\vert\lesssim\lambda^{\alpha_K^U}(|\phi(p)-\phi(q)|+\lambda)^{\gamma-\alpha_K^U}\,,
\end{align}
uniformly for $p,q\in K$, $\lambda\in(0,1]$.
\end{definition}

\begin{remark}
At first sight, the above definition depends on the atlas $\mathcal{A}$. Nonetheless in Proposition \ref{Prop: independence from the atlas} we shall prove that the above definition is actually independent of the atlas.
\end{remark}

\begin{remark}\label{Rmk: who cares about 1}
In the previous definition, we adopted the constraint $\lambda\in(0,1]$. Nonetheless, this can be replaced by $\lambda\in(0,\eta]$, for any $\eta>0$. Indeed, all bounds are given up to a multiplicative constant. We shall use this fact in the following when discussing enhanced coherence in Appendix \ref{Appendix: enhanced coherence on an open set}. Moreover, in the following we shall be interested in the behaviour of all structures for $\lambda\to0^+$. These are not influenced by the choice of $\eta>0$.
\end{remark}

First of all, we can refine the dependence on the atlas of the notion of coherence. In particular, it is independent of the coordinates.
\begin{proposition}\label{Proposition: independence on charts}
With the notation of Definition \ref{Def: coherence}, let $F$ be a $\gamma$-coherent germ on $(M,\mathcal{A})$ and let $(U,\phi)\in\mathcal{A}$. Let $(U,\psi)$ be a second chart on the same open set $U\subset M$. Then the $\gamma$-coherence condition holds true also with respect to the chart $(U,\psi)$. Moreover, also the $\bm{\alpha}^U$ parameters are independent of the coordinates.
\end{proposition}
\begin{proof}
First of all, coherence on $(M,\mathcal{A})$, entails the existence of a test-function $f\in\mathcal{D}(\phi(U))$ with $\int_{\mathbb{R}^d}dx\,f(x)\neq0$ such that for any compact set $K\subset U$ there exists $\alpha_K^U\leq\min\{0,\gamma\}$ for which
\begin{align*}
\vert(\phi_*(F_p)-\phi_*(F_q))(f^\lambda_{\phi(q)})\vert\lesssim\lambda^{\alpha_K^U}(|\phi(p)-\phi(q)|+\lambda)^{\gamma-\alpha_K^U}\,,
\end{align*}
uniformly for $p,q\in K$, and for $\lambda\in(0,1]$.
We prove that there exists a test-function $\tilde{g}\in\mathcal{D}(\psi(U))$ such that the coherence condition with respect to the chart $(U,\psi)$ is satisfied. For any $g\in\mathcal{D}(\psi(U))$, it holds
\begin{align*}
\vert(\psi_*(F_p)-\psi_*(F_q))(g^\lambda_{\psi(q)})\vert&=\vert((\psi\circ\phi^{-1})_*\phi_*(F_p)-(\psi\circ\phi^{-1})_*\phi_*(F_q))(g^\lambda_{\psi(q)})\vert\\
&\lesssim \vert(\phi_*(F_p)-\phi_*(F_q))((\psi\circ\phi^{-1})^*g)^\lambda_{\phi(q)}\vert\,,
\end{align*}
where  the last inequality descends from $\|\mathrm{Jac}(\psi\circ\phi^{-1})\|_\infty\lesssim1$, where $\|\cdot\|_\infty$ denotes the supremum norm and $\mathrm{Jac}(\psi\circ\phi^{-1})$ denotes the Jacobian of coordinates change. We can now choose $\tilde{g}\in\mathcal{D}(\psi(U))$ such that $(\psi\circ\phi^{-1})^*\tilde{g}=f$, with $f\in\mathcal{D}(\phi(U))$ as above. As a consequence, we get 
\begin{align*}
\vert(\psi_*(F_p)-\psi_*(F_q))(\tilde{g}^\lambda_{\psi(q)})\vert\lesssim\lambda^{\alpha_K^U}(|\phi(p)-\phi(q)|+\lambda)^{\gamma-\alpha_K^U}\lesssim\lambda^{\alpha_K^U}(|\psi(p)-\psi(q)|+\lambda)^{\gamma-\alpha_K^U}\,,
\end{align*}
where in the last inequality we exploited the uniform bound 
\begin{align*}
\sup_{\overset{p,q\in K}{p\neq q}}\frac{|\phi(p)-\phi(q)|}{|\psi(p)-\psi(q)|}\lesssim1\,.
\end{align*}
\end{proof}

\begin{remark}\label{Remark: equivalent definition of coherence}
The notion of coherence of a germ of distributions can be stated in an equivalent form by splitting the two cases $|\phi(p)-\phi(q)|\leq\lambda$ and $|\phi(p)-\phi(q)|>\lambda$. In particular, Equation \eqref{Eq: coherence condition} can be rewritten as
\begin{align}\label{Eq: equivalent definition of coherence}
\vert(\phi_*(F_p)-\phi_*(F_q))(f^\lambda_{\phi(q)})\vert\lesssim 
 \begin{cases}
\lambda^{\gamma}\,,&\qquad\mathrm{if}\quad0\leq |\phi(p)-\phi(q)|\leq\lambda\,, \\
\lambda^{\alpha^U_K}|\phi(p)-\phi(q)|^{\gamma-\alpha_K^U}\,,&\qquad\mathrm{if}\quad |\phi(p)-\phi(q)|>\lambda\,, \\
   \end{cases}
\end{align}
\end{remark}
\paragraph{Enhanced Coherence}
In this paragraph we shall refine the notion of coherence on a smooth manifold. This leads to the notion of \emph{enhanced coherence}. 
In the same spirit of \cite{CZ20}, the idea is to drop the dependence on the particular test function $f\in\mathcal{D}(\phi(U))$.
To this end, we resort to the same argument for the case of an open subset of $\mathbb{R}^d$.
Indeed, on account of Definition \ref{Def: coherence}, given a $\gamma$-coherent germ $F_p$ on a smooth manifold $(M,\mathcal{A})$ and a local chart $(U,\phi)\in\mathcal{A}$, $\mathcal{F}_{\phi(p)}\vcentcolon=\phi_*(F_p)$ is a $\gamma$-coherent germ of distributions on the open set $\phi(U)\subset\mathbb{R}^d$, in the sense of Definition \ref{Def:coherence on an open set}.
As a consequence, we can apply \emph{locally} Proposition \ref{Prop:enhanced_coherence on an open set} to get the following definition of coherence on a smooth manifold, which is equivalent to Definition \ref{Def: coherence} -- \textit{cf}. Appendix \ref{Appendix: enhanced coherence on an open set}.

\begin{definition}\label{Def: coherence final version}
Let $M$ be a smooth $d$-dimensional manifold and let $\mathcal{A}=\{(U_\alpha,\phi_\alpha)\}_{\alpha}$ be a smooth atlas on $M$.
Let $F=(F_p)_{p\in M}$ be a germ of distributions on $M$ and let $\gamma\in\mathbb{R}$. We say that $F$ is $\gamma$-coherent on $(M,\mathcal{A})$ if for any $(U,\phi)\in\mathcal{A}$ and for any $K\subset U$ compact there exists $\alpha_K^U\leq\min\{0,\gamma\}$ such that, for any $r>-\alpha_K^U$
\begin{align}\label{Eq: enhanced coherence}
|(\phi_*(F_p)-\phi_*(F_q))(u_{\phi(q)}^\lambda)|\lesssim\|u\|_{C^r}\lambda^{\alpha_K^U}(|\phi(p)-\phi(q)|+\lambda)^{\gamma-\alpha_K^U}\,,
\end{align}
uniformly for $p,q\in K$, $\lambda\in(0,1]$ and $u\in\mathcal{D}(B(0,1))$, where $B(0,1)\subset\mathbb{R}^d$ denotes the unitary ball centred at the origin. 
\end{definition}

\begin{remark}
Although Definition \ref{Def: coherence} and Definition \ref{Def: coherence final version}  are equivalent, the latter is more advantageous. Indeed, it establishes a bound which is independent of the test function. As a by product, it also provides the space of $\gamma$-coherent germs of distributions with a vector space structure. At the same time, Definition \ref{Def: coherence} is preferable from a computational point of view, since it allows to establish coherence by only checking the defining property for a single test-function.
\end{remark}

\begin{remark}
Observe that the equivalence between Definition \ref{Def: coherence} and Definition \ref{Def: coherence final version} entails that also the notion of enhanced coherence is independent of the coordinates -- see Proposition \ref{Proposition: independence on charts}. Alternatively, this independence can be proven directly following an approach similar to that of Proposition \ref{Proposition: independence on charts} and exploiting the boundedness property of the Jacobian of the change of coordinates. 
\end{remark}

\begin{remark}\label{Rmk: coherence is stable wrt restriction}
On account of Proposition \ref{Prop:restriction on an open set}, also in the case of a smooth manifold, the notion of coherence is stable under restriction of an open set. More precisely, adopting the same notation of Definition \ref{Def: coherence final version}, if we consider $V\subset U$, then $F_p$ is a $\gamma$-coherent germ of distributions in the sense of Definition \ref{Def: coherence final version} also with respect to any local chart $(V,\phi)$. 
\end{remark}
We are now in position to prove that coherence is independent of the atlas.

\begin{proposition}\label{Prop: independence from the atlas}
Let $M$ be a smooth $d$-dimensional manifold and let $\mathcal{A}$ and $\mathcal{A}^\prime$ be smooth atlases on $M$. Let $F=\{F_p\}_{p\in M}$ be a germ of distributions over $M$. If $F=\{F_p\}_{p\in M}$ is $\gamma$-coherent with respect to $(M,\mathcal{A})$, then it is $\gamma$-coherent also with respect to $(M,\mathcal{A}^\prime)$.
\end{proposition}
\begin{proof}
In order to prove this result, on account of the definition of coherence, it suffices to prove that for any $(U^\prime,\phi^\prime)\in\mathcal{A}^\prime$, $F$ satisfies the bound of Equation \eqref{Eq: enhanced coherence} on $U^\prime$.
Notice moreover that we can focus on the open set $U^\prime$, \textit{i.e.}, neglecting the local chart $\phi^\prime$, since in Proposition \ref{Proposition: independence on charts} we proved that coherence is independent of the coordinates. 
Hence, let $U^\prime\in\mathcal{A}^\prime$. 
There exists a family $\{U_i\}_{i\in I}\subset\mathcal{A}$ of open sets such that $U^\prime=\bigcup_{i\in I}U_i\cap U^\prime=\bigcup_{i\in I}U_i^\prime$, where we set $U_i^\prime\vcentcolon=U_i\cap U^\prime$. 
Notice that by independence of the coordinate, \textit{cf.} Proposition  \ref{Proposition: independence on charts}, and by stability of coherence under restriction of the open set, \textit{cf.} Remark \ref{Rmk: coherence is stable wrt restriction}, $F_p$ satisfies the bound of $\gamma$-coherence, given by Equation \eqref{Eq: enhanced coherence}, on all the open sets $U^\prime_i$, for $i\in I$. 
Moreover, all these sets being contained in $U^\prime$, we can set on all of them a unique coordinate, which we call $\phi$. 
In order to prove the thesis, we first prove the following claim: the coherence bound of Equation \eqref{Eq: enhanced coherence} holds true on the union of two open sets $U^\prime_j$ and $U^\prime_\ell$, with $U^\prime_j\cap U^\prime_\ell\neq\emptyset$. \\
To this end, let $K\subset U^\prime_j\cup U^\prime_\ell$ be a compact set. Notice that if the compact set $K$ is contained in one of the two open sets $U^\prime_j$ or $U^\prime_\ell$, then the thesis holds true by construction. 
It remains to consider only the case of a compact set $K\subset U^\prime_j\cup U^\prime_\ell$ such that $K\cap U^\prime_k\neq\emptyset$ and $K\cap U^\prime_j\neq\emptyset$. \\
In this scenario, we can split the compact set as $K=K_j\cup K_\ell$, with $K_j\subset U^\prime_j$ and $K_\ell\subset U^\prime_\ell$ compact sets and $K_j\cap K_\ell\neq\emptyset$.
 In the next step, we shall prove the coherence bound, Equation \eqref{Eq: enhanced coherence}, uniformly on $p,q\in K$. 
 Notice that whether this two points were both contained in one of the two compact set $K_j$ or $K_\ell$, then the proof is already complete since these two compact sets are contained in $U^\prime_j$ and $U^\prime_\ell$ respectively. 
 As a consequence, it only remains to discuss the case with $p\in K_j\setminus U^\prime_\ell$ and $q\in K_\ell\setminus U^\prime_j$.\\
For any $u\in\mathcal{D}(B(0,1))$, by triangular inequality,
\begin{align}\label{Eq: first estimate}
|(\phi_*(F_p)-\phi_*(F_q))(u^\lambda_{\phi(q)})|\leq
\underset{|A|}{\underbrace{|(\phi_*(F_p)-\phi_*(F_a))(u^\lambda_{\phi(q)})|}}+\underset{|B|}{\underbrace{|(\phi_*(F_a)-\phi_*(F_q))(u^\lambda_{\phi(q)})|}}\,,
\end{align}
for any $a\in K_j\cap K_\ell$. 
Moreover, we fix $r>\max\big\{-\alpha_{K_j}^{U^\prime_j},-\alpha_{K_\ell}^{U^\prime_\ell}\big\}$. 
We separately estimate $|A|$ and $|B|$. First, on account of the choice of $r$ and of $a$, we have, by Equation \eqref{Eq: enhanced coherence} on $U^\prime_\ell$,
\begin{align*}
|B|\lesssim\|u\|_{\mathrm{C}^r}\lambda^{\alpha_{K_\ell}^{U^\prime_\ell}}(|\phi(a)-\phi(q)|+\lambda)^{\gamma-\alpha_{K_\ell}^{U^\prime_\ell}}\,,
\end{align*}
uniformly for $a,q\in K_\ell$ and for $\lambda\in(0,1]$. Moreover, notice that as a consequence of the estimate
\begin{align}\label{Eq: general estimate}
\sup_{\overset{\lambda\in(0,1]}{a\in K_j\cap K_\ell}}\sup_{\overset{p\in K_j\setminus U^\prime_\ell}{q\in K_\ell\setminus U^\prime_j}}\frac{(|\phi(a)-\phi(q)|+\lambda)^{\gamma-\alpha_{K_\ell}^{U^\prime_\ell}}}{(|\phi(p)-\phi(q)|+\lambda)^{\gamma-\alpha_{K_\ell}^{U^\prime_\ell}}}\lesssim 1\,,
\end{align}
we get
\begin{align}\label{Eq: bound for |B|}
|B|\lesssim\|u\|_{\mathrm{C}^r}\lambda^{\alpha_{K_\ell}^{U^\prime_\ell}}(|\phi(p)-\phi(q)|+\lambda)^{\gamma-\alpha_{K_\ell}^{U^\prime_\ell}}\,.
\end{align}
The estimate for $|A|$ requires some more steps: first of all, we notice that in $|A|$ the test-function is centred at $\phi(q)$. Nonetheless, we can center it at the point $\phi(a)$ by exploiting the argument used in the proof of \cite[Prop. 6.2]{CZ20}. This is achieved by noticing that 
\begin{align*}
u^\lambda_{\phi(q)}=\tilde{u}^{\lambda_1}_{\phi(a)}\,,\qquad\mathrm{with}\qquad\tilde{u}\vcentcolon=u^{\lambda_2}_w\,,
\end{align*}
with
\begin{align*}
\lambda_1=|\phi(q)-\phi(a)|+\lambda\,,\qquad\lambda_2=\frac{\lambda}{\lambda_1}\,,\qquad w=\frac{\phi(q)-\phi(a)}{|\phi(q)-\phi(a)|+\lambda}\,.
\end{align*}
On account of this and of the coherence on $U_j^\prime$, we get
\begin{align*}
|A|\lesssim\|\tilde{u}\|_{\mathrm{C}^r}\lambda^{\alpha_{K_j}^{U^\prime_j}}(|\phi(p)-\phi(a)|+\lambda)^{\gamma-\alpha_{K_j}^{U^\prime_j}}\,.
\end{align*}
By definition of $\tilde{u}$,
\begin{align*}
\|\tilde{u}\|_{\mathrm{C}^r}\lesssim\lambda_2^{-r-d}\|u\|_{\mathrm{C}^r}\lesssim\lambda^{-r-d}\|u\|_{\mathrm{C}^r}\,.
\end{align*}
Hence, we get
\begin{align*}
|A|\lesssim\|u\|_{\mathrm{C}^r}\lambda^{\alpha_{K_j}^{U^\prime_j}-r-d}(|\phi(p)-\phi(a)|+\lambda)^{\gamma-\alpha_{K_j}^{U^\prime_j}}\lesssim\|u\|_{\mathrm{C}^r}\lambda^{\widetilde{\alpha}_{K_j}^{U^\prime_j}}(|\phi(p)-\phi(a)|+\lambda)^{\gamma-\widetilde{\alpha}_{K_j}^{U^\prime_j}}\,,
\end{align*}
where, in the last inequality we set $\widetilde{\alpha}_{K_j}^{U^\prime_j}\vcentcolon=\alpha_{K_j}^{U^\prime_j}-r-d$ and where, always in the last inequality, we used 
\begin{align*}
\sup_{\overset{\lambda\in(0,1], p\in K_j\setminus U^\prime_\ell}{a\in K_j\cap K_\ell}}(\lambda+|\phi(p)-\phi(a)|)^{-r-d}\lesssim 1\,.
\end{align*}
With a bound analogous to Equation \eqref{Eq: general estimate}, we conclude
\begin{align}\label{Eq: bound for |A|}
|A|\lesssim\|u\|_{\mathrm{C}^r}\lambda^{\widetilde{\alpha}_{K_j}^{U^\prime_j}}(|\phi(p)-\phi(q)|+\lambda)^{\gamma-\widetilde{\alpha}_{K_j}^{U^\prime_j}}\,.
\end{align}
Finally, on account of Equations \eqref{Eq: first estimate}, \eqref{Eq: bound for |B|} and \eqref{Eq: bound for |A|}, setting $\alpha_K^{U^\prime_j\cup U_\ell^\prime}\vcentcolon=\min\{\alpha_{K_\ell}^{U_\ell^\prime}, \widetilde{\alpha}_{K_j}^{U_j^\prime}\}$, we get, for any $r>-\alpha_K^{U^\prime_j\cup U_\ell^\prime}$,
\begin{align*}
|(\phi_*(F_p)-\phi_*(F_q))(u^\lambda_{\phi(q)})|\lesssim\|u\|_{\mathrm{C}^r}\lambda^{\alpha_K^{U^\prime_j\cup U_\ell^\prime}}(|\phi(p)-\phi(q)|+\lambda)^{\gamma-\alpha_K^{U^\prime_j\cup U_\ell^\prime}}\,,
\end{align*}
uniformly on $p,q\in K$. This concludes the proof of the claim. In order to conclude the proof of the proposition, we distinguish two cases. On the one hand, if the open set $U^\prime$ is bounded, then the proof is complete since $U^\prime$ can be covered by a finite number of open sets $U_i\in\mathcal{A}$ and it suffices to iterate the above procedure for a finite number of times. On the other hand, if $U^\prime$ is unbounded, it suffices to notice that for any compact set $K\subset U^\prime$, there exists a finite subset $J\subset I$ of indices such that $K\subset\cup_{j\in J}U^\prime_j$. As a consequence, we can get a coherence parameter $\alpha_K^{U^\prime}$ for $K$ by iterating the above claim a finite number of times. 
\end{proof}

We give two simple examples of coherent germs on a smooth manifold.
\begin{example}
A simple example of coherent germ of distributions on a smooth manifold $M$ is the following. Consider a distribution $t\in\mathcal{D}^\prime(M)$ and set $F_p\vcentcolon=t$ for any $p\in M$. Since $F_p-F_q=0$ for any $p,q\in M$, we conclude that, on any $U\subset M$, $\{F_p\}_{p\in M}$ is coherent with any parameters $(\bm{\alpha}^U,\gamma)$.
\end{example}
\begin{example}
Notice that our construction is a generalization of the one of \cite{CZ20}: indeed, we recover their construction if we consider the case $M=\mathbb{R}^d$ endowed of the trivial atlas $(\mathbb{R}^d,\mathrm{Id})$. As a consequence, all examples discussed in \cite{CZ20}, such as Taylor polynomials, are coherent germs with respect to this atlas.
\end{example}

\begin{comment}
\begin{remark}\label{Example: euclidean case}
We underline that, being the coherence of a germ of distributions on a manifold a notion stable with respect to restriction, \textit{cf}. Remark \ref{Rmk: coherence is stable wrt restriction}, and independent of the coordinates, \textit{cf}. Proposition \ref{Proposition: independence on charts}, coherence with respect to the trivial atlas $(\mathbb{R}^d,\mathrm{Id})$ implies that on all the possible atlases over the manifold $M=\mathbb{R}^d$, \textit{cf}. Remark \ref{Rmk: different atlases} for some more details about this topic.
\end{remark}
\end{comment}

Eventually, we introduce a \emph{homogeneity} parameter for coherent germs of distributions.
\begin{lemma}\label{Lemma: homogeneity}
Let $F=(F_p)_{p\in M}$ be a $\gamma$-coherent germ of distributions on a smooth manifold $(M,\mathcal{A})$ and let $(U,\phi)\in\mathcal{A}$ be a local chart. For any compact set $K\subset U$, there exists $\beta_K^U<\gamma$ such that
\begin{align}\label{Eq: homogeneity condition}
\vert \phi_*(F_p)(f_{\phi(p)}^\lambda)\vert\lesssim\lambda^{\beta_K^U}\,,
\end{align}
uniformly for $p\in K$, $\lambda\in(0,1]$ and where $f\in\mathcal{D}(\phi(U))$ is chosen as in Definition \ref{Def: coherence}. We say that $F$ is \emph{locally homogeneous} in $U$ with parameters $\bm{\beta}^U=(\beta_K^U)_K$.
\end{lemma}
\begin{proof}
The proof is analogous to that on $\mathbb{R}^d$ \cite{CZ20}.
First of all, given a compact set $K\subset U$ and $q\in K$, since $\phi_*(F_q)\in\mathcal{D}^\prime(\phi(U))$, there exists, \textit{cf}. \cite[Remark 3.5]{CZ20}, $r\in\mathbb{N}$ such that
\begin{align*}
\vert \phi_*(F_q)(f_{\phi(p)}^\lambda)\vert\lesssim\lambda^{-d-r}\,,
\end{align*}
uniformly for $p\in K$ and for $\lambda\in(0,1]$. Moreover, since $\mathrm{Diam}(\phi(K))\vcentcolon=\sup_{p,q\in K}|\phi(p)-\phi(q)|<\infty$, on account of the coherence condition it holds, uniformly for $p,q\in K$,
\begin{align*}
\vert(\phi_*(F_p)-\phi_*(F_q))(f^\lambda_{\phi(q)})\vert\lesssim\lambda^{\alpha_K^U}(|\phi(p)-\phi(q)|+\lambda)^{\gamma-\alpha_K^U}\leq\lambda^{\alpha_K^U}(\mathrm{Diam}(\phi(K))+\lambda)^{\gamma-\alpha_K^U}\lesssim\lambda^{\alpha_K^U}\,.
\end{align*}
In addition, thanks to the triangular inequality it descends
\begin{align*}
\vert \phi_*(F_p)(f_{\phi(p)}^\lambda)\vert\leq\vert \phi_*(F_q)(f_{\phi(p)}^\lambda)\vert+\vert(\phi_*(F_p)-\phi_*(F_q))(f_{\phi(p)}^\lambda)\vert\lesssim\lambda^{-d-r}+\lambda^{\alpha_K^U}\lesssim\lambda^{\min\{-d-r,\alpha_K^U\}}\,.
\end{align*}
As a consequence, it suffices to choose any $\beta_K^U\leq\min\{-d-r,\alpha_K^U\}$. Eventually, we can choose $\beta_K^U\leq\min\{-d-r,\alpha_K^U,\gamma\}$. This concludes the proof.
\end{proof}
Similarly to the case of the coherence parameters, also those related to homogeneity are independent of the coordinates.

\begin{proposition}
Let $(U,\phi)\in\mathcal{A}$ be a local chart on a smooth manifold $M$ and let $F=(F_p)_{p\in M}$ be a $\gamma$-coherent germ with respect to $(M,\mathcal{A})$. Moreover let $f\in\mathcal{D}(\phi(U))$ be the test-function as per Definition \ref{Def: coherence}. Let $\bm{\beta}^U=(\beta_K^U)_K$ be the homogeneity parameters of $F$, as per Lemma \ref{Lemma: homogeneity}. Let $(U,\psi)$ be a local chart on the same open set $U\subset M$. Then the homogeneity condition, \textit{cf}. Equation \eqref{Eq: homogeneity condition}, holds true also with respect to $(U,\psi)$.
\end{proposition}
\begin{proof}
Let $K\subset U$ be a compact set and let $(U,\phi)$ ad $(U,\psi)$ be as above. It holds, uniformly for $p\in K$ and $\lambda\in(0,1]$,
\begin{align*}
\vert \psi_*(F_p)(((\phi\circ\psi^{-1})^*f)_{\psi(p)}^\lambda)\vert \lesssim \lvert \phi_*(F_p)(f_{\phi(p)}^\lambda) \rvert \lesssim \lambda^{\beta^U_K}\,,
\end{align*}
where in the first inequality we exploited $\|\mathrm{Jac}(\psi\circ\phi^{-1})\|_\infty\lesssim1$, whereas in the last inequality we used the homogeneity hypothesis with respect to the chart $(U,\phi)$. It follows that $\psi_*(F_p)$ has homogeneities $(\beta_K^U)$ with respect to the test function $(\phi\circ\psi^{-1})^*f\in\mathcal{D}(\psi(U))$.
\end{proof}

\section{Reconstruction Theorem}\label{Sec: Reconstruction theorem}
We are now in position to state the main result of the paper, the \emph{reconstruction theorem} on a smooth manifold, for $\gamma$-coherent germs of distributions with $\gamma>0$, \textit{cf}. Theorem \ref{Thm: main result}.
\begin{theorem}\label{Thm: reconstruction}
Let $M$ be a $d$-dimensional smooth manifold and let $\mathcal{A}=\{(U_j,\phi_j)\}_{j}$ be an atlas over $M$. Let $\gamma>0$ and let $F=(F_p)_{p\in M}$ be a $\gamma$-coherent germ of distributions on $(M,\mathcal{A})$. There exists a unique distribution $\mathcal{R}F\in\mathcal{D}^\prime(M)$ such that, for any $(U,\phi)\in\mathcal{A}$, $\phi_*(\mathcal{R}F)\in\mathcal{D}^\prime(\phi(U))$ and it satisfies, for any compact set $K\subset U$ and for any $h\in\mathcal{D}(\phi(U))$,
\begin{align}\label{Eq: reconstruction}
|(\phi_*(\mathcal{R}F)-\phi_*(F_p))(h_{\phi(p)}^\lambda)|\lesssim\lambda^\gamma\,,
\end{align}
uniformly for $p\in K$ and $\lambda\in(0,1]$.
\end{theorem}
\begin{proof}
The proof of this result is mainly based on the application of two theorems, namely \cite[Theor. 4.4]{CZ20} and Theorem \ref{Trm: Hormander}. 
First of all, on account of Definition \ref{Def: coherence}, for any $(U,\phi)\in\mathcal{A}$, $\mathcal{F}_{\phi(p)}\vcentcolon=\phi_*(F_p)$ is a $\gamma$-coherent germ on the open set $\phi(U)\subset\mathbb{R}^d$, with $\gamma>0$. 
As a consequence, \cite[Theor. 4.4]{CZ20} implies the existence of a unique distribution $(\mathcal{R}F)^{\phi(U)}\in\mathcal{D}^\prime(\phi(U))$ such that, for any compact set $K\subset U$ and for any $h\in\mathcal{D}(\phi(U))$
\begin{align}\label{Eq: defining reconstruction}
|((\mathcal{R}F)^{\phi(U)}-\phi_*(F_p))(h_{\phi(p)}^\lambda)|\lesssim\lambda^\gamma\,,
\end{align}
uniformly for $p\in K$ and for $\lambda\in(0,1]$. 
Thus we have a family of local distributions $(\mathcal{R}F)^{\phi_j(U_j)}\in\mathcal{D}^\prime(\phi_j(U_j))$, labelled by the local charts within the atlas $\mathcal{A}$. 
Due to Theorem \ref{Trm: Hormander}, this family identifies a unique distribution, $\mathcal{R}F\in\mathcal{D}^\prime(M)$, such that $\phi_*(\mathcal{R}F)=(\mathcal{R}F)^{\phi(U)}$ for any $(U,\phi)\in\mathcal{A}$, if and only if
\begin{align}\label{Eq: overlapping condition}
(\mathcal{R}F)^{\phi(U)}=(\psi\circ\phi^{-1})^*(\mathcal{R}F)^{\psi(V)}\qquad\mathrm{on}\quad\phi(U\cap V)\,,
\end{align}
for any $(U,\phi),(V,\psi)\in\mathcal{A}$. To this end, we fix a compact $K\subset U\cap V$ and we notice that for any $g\in\mathcal{D}(\phi(U\cap V))$,
\begin{align*}
|((\mathcal{R}F)^{\phi(U)}-(\psi\circ\phi^{-1})^*&(\mathcal{R}F)^{\psi(V)})(g_{\phi(p)}^\lambda)|\\
\leq&\,\underset{|A|}{\underbrace{|((\mathcal{R}F)^{\phi(U)}-\phi_*(F_p))(g_{\phi(p)}^\lambda)|}}+\underset{|B|}{\underbrace{|((\psi\circ\phi^{-1})^*(\mathcal{R}F)^{\psi(V)}-\phi_*(F_p))(g_{\phi(p)}^\lambda)|}}\,.
\end{align*}
By construction, on the one hand $|A|\lesssim\lambda^\gamma$ uniformly on $p\in K$ and $\lambda\in(0,1]$. On the other hand, again uniformly on $p\in K$ and for $\lambda\in(0,1]$,
\begin{align*}
|B|=|((\phi\circ\psi^{-1})_*(\mathcal{R}F)^{\psi(V)}-\phi_*(F_p))(g_{\phi(p)}^\lambda)|\lesssim|((\mathcal{R}F)^{\psi(V)}-\psi_*(F_p))((\phi\circ\psi^{-1})^*g)_{\psi(p)}^\lambda|\lesssim\lambda^\gamma\,,
\end{align*}
where in the first inequality we exploited $\|\mathrm{Jac}(\phi\circ\psi^{-1})\|_\infty\lesssim1$, whereas in the last inequality we used the defining inequality, Equation \eqref{Eq: defining reconstruction}, of $(\mathcal{R}F)^{\psi(V)}$ together with $(\phi\circ\psi^{-1})^*g\in\mathcal{D}(\psi(V))$. Summarizing,
\begin{align*}
|((\mathcal{R}F)^{\phi(U)}-(\psi\circ\phi^{-1})^*(\mathcal{R}F)^{\psi(V)})(g_{\phi(p)}^\lambda)|\lesssim\lambda^\gamma\,,
\end{align*}
uniformly on $p\in K$ and $\lambda\in(0,1]$. Finally, exploiting $\gamma>0$,
\begin{align*}
|((\mathcal{R}F)^{\phi(U)}-(\psi\circ\phi^{-1})^*(\mathcal{R}F)^{\psi(V)})(g_{\phi(p)}^\lambda)|\to0\qquad\mathrm{for}\quad\lambda\to0^+\,.
\end{align*}
Applying Lemma \ref{Lemma: partial evaluation} to the distribution $T\vcentcolon=(\mathcal{R}F)^{\phi(U)}-(\psi\circ\phi^{-1})^*(\mathcal{R}F)^{\psi(V)}$ on the open set $\phi(U\cap V)$, we get
\begin{align*}
(\mathcal{R}F)^{\phi(U)}=(\psi\circ\phi^{-1})^*(\mathcal{R}F)^{\psi(V)}\qquad\mathrm{on}\quad\phi(U\cap V)\,.
\end{align*}
This concludes the proof.
\end{proof}

\begin{remark}
On account of a local version of \cite[Thm. 12.7]{CZ20} and of the \emph{homogeneity} of germs of distributions -- see Lemma \ref{Lemma: homogeneity}, one can conclude that the reconstructed distributions $\{(\mathcal{R}F)^{\phi(U)}\}_{(U,\phi)}$ appearing in the proof of the reconstruction theorem are actually elements of $\mathcal{C}^{\beta^U}$, which is a \emph{H\"older space of negative regularity} \cite[Section 12]{CZ20}. As a consequence of our construction of Theorem \ref{Thm: reconstruction} we conclude that locally the distribution $\mathcal{R}F\in\mathcal{D}^\prime(M)$ is of regularity $\mathcal{C}^{\beta^U}$.
\end{remark}

In the next theorem we investigate in detail the dependence of the \emph{reconstructed distribution} on the atlas. In particular, we shall prove that given a germ of distributions which is $\gamma$-coherent, with $\gamma>0$, then the reconstruction is independent of the atlas.
\begin{theorem}\label{Thm: coherence wrt different atlases}
Let $M$ be a $d$-dimensional smooth manifold and let $\mathcal{A}$ and $\mathcal{A}^\prime$ be two atlases over $M$. Let $\gamma>0$ and let $F=\{F_p\}_{p\in M}$ be a $\gamma$-coherent germ of distributions. Denote with $\mathcal{R}_\mathcal{A}F\in\mathcal{D}^\prime(M)$ and $\mathcal{R}_{\mathcal{A}^\prime}F\in\mathcal{D}^\prime(M)$ the reconstructed distributions associated with the germ $F$ with respect to $\mathcal{A}$ and $\mathcal{A}^\prime$, as per Theorem \ref{Thm: reconstruction}. Then $\mathcal{R}_{\mathcal{A}}F=\mathcal{R}_{\mathcal{A}^\prime}F$, \textit{i.e.}, the reconstruction is independent of the atlas.
\end{theorem}
\begin{proof}
Exploiting Theorem \ref{Thm: reconstruction}, we can associate with the germ $F$ and with the atlas $(M,\mathcal{A})$ the global distribution $\mathcal{R}_{\mathcal{A}}F\in\mathcal{D}^\prime(M)$, which is identified by means of Theorem \ref{Trm: Hormander} by the family $\{(\mathcal{R}_{\mathcal{A}}F)^{\phi(U)}\}_{(U,\phi)\in\mathcal{A}}$, with $(\mathcal{R}_{\mathcal{A}}F)^{\phi(U)}\in\mathcal{D}^\prime(\phi(U))$. 
Moreover, let $(U^\prime,\phi^\prime)\in\mathcal{A}^\prime$ be a local chart in the atlas $\mathcal{A}^\prime$.
 On the one hand, we can introduce the distribution $\phi^\prime_*(\mathcal{R}_{\mathcal{A}}F)\in\mathcal{D}^\prime(\phi^\prime(U^\prime))$ while, on the other hand, as a consequence of Theorem \ref{Thm: reconstruction} applied with reference to the atlas $\mathcal{A}^\prime$, it descends that the distribution $(\mathcal{R}_{\mathcal{A}^\prime}F)^{\phi^\prime(U^\prime)}\in\mathcal{D}^\prime(\phi^\prime(U^\prime))$. 
 Recalling the notion of distribution on a manifold, \textit{cf}. Appendix \ref{Appendix: distribution on manifold}, to conclude the proof of this theorem it suffices to show that 
\begin{align}\label{Eq: new claim}
\phi^\prime_*(\mathcal{R}_{\mathcal{A}}F)=(\mathcal{R}_{\mathcal{A}^\prime}F)^{\phi^\prime(U^\prime)}\,,\qquad\mathrm{in}\quad\mathcal{D}^\prime(\phi^\prime(U^\prime))\,.
\end{align}
To this end, there exists a family of local charts $\{(U_i,\phi_i)\}_{i\in I}\subset\mathcal{A}$, for some index set $I$, such that $U^\prime=\cup_{i\in I}(U^\prime\cap U_i)=\vcentcolon\cup_{i\in I}U^\prime_i$, with $U^\prime_i\vcentcolon=U^\prime\cap U_i$. 
We consider the restriction of Equation \eqref{Eq: new claim} to a subset $U^\prime_i$.
Recalling that coherence is stable with respect to restrictions, entailing by uniqueness $(\mathcal{R}_{\mathcal{A}^\prime}F)^{\phi^\prime(U^\prime)}|_{\phi^\prime(U^\prime_i)}=(\mathcal{R}_{\mathcal{A}^\prime}F)^{\phi^\prime(U^\prime_i)}$, we prove that, for any $i\in I$,
\begin{align}\label{Eq: latest claim}
\phi^\prime_*(\mathcal{R}_{\mathcal{A}}F|_{U^\prime_i})=(\mathcal{R}_{\mathcal{A}^\prime}F)^{\phi^\prime(U^\prime_i)}\,.
\end{align}
Via a partition of unity argument, this yields Equation \eqref{Eq: new claim}. For any compact set $K\subset U^\prime_i$ and for any $h\in\mathcal{D}(\phi^\prime(U^\prime_i))$, it holds, uniformly on $p\in K$ and $\lambda\in(0,1]$,
\begin{align}\label{Eq: check of coherence}
|(\phi^\prime_*(\mathcal{R}_{\mathcal{A}}F|_{U^\prime_i})-\phi_*^\prime(F_p))(&h_{\phi^\prime(p)}^\lambda)|\nonumber\\
\lesssim&\,|(\phi_{i*}(\mathcal{R}_{\mathcal{A}}F|_{U^\prime_i})-\phi_{i*}(F_p))((\phi^\prime\circ\phi_i^{-1})^*h)_{\phi(p)}^\lambda|\lesssim\lambda^\gamma\,,
\end{align}
where in the first inequality we performed a change of coordinates whereas in the last inequality we exploited that $\mathcal{R}_\mathcal{A}F$ reconstructs $F$ with respect to the atlas $\mathcal{A}$. 
Finally, we recall that on account of Theorem \ref{Thm: reconstruction} and of $\gamma>0$, $(\mathcal{R}_{\mathcal{A}^\prime}F)^{\phi^\prime(U^\prime_i)}$ is the unique distribution satisfying Equation \eqref{Eq: check of coherence}. 
Hence, Equation \eqref{Eq: latest claim} holds true by uniqueness. 
This concludes the proof.
\end{proof}

Eventually, we discuss the reconstruction theorem for the case of $\gamma$-coherent germs of distributions with $\gamma\leq0$.
\begin{theorem}\label{Rmk: gamma<0}
Let $M$ be a $d$-dimensional smooth manifold and let $\mathcal{A}=\{(U_j,\phi_j)\}_{j\in J}$ be an atlas over $M$, with $J$ index set. Let $\gamma\leq0$ and let $F=(F_p)_{p\in M}$ be a $\gamma$-coherent germ of distributions on $(M,\mathcal{A})$. There exists a distribution $\mathcal{R}F\in\mathcal{D}^\prime(M)$ such that, for any $(U,\phi)\in\mathcal{A}$, $\phi_*(\mathcal{R}F)\in\mathcal{D}^\prime(\phi(U))$ and it satisfies, for any compact set $K\subset U$ and for any $h\in\mathcal{D}(\phi(U))$,
\begin{equation}\label{Eq: reconstruction for gamma le0}
|(\phi_*(\mathcal{R}F)-\phi_*(F_p))(h_{\phi(p)}^\lambda)|\lesssim
\begin{cases}
\lambda^\gamma\qquad\mathrm{if}\;\gamma<0\,,\\
1+|\log\lambda|\qquad\mathrm{if}\;\gamma=0\,,
\end{cases}
\end{equation}
uniformly for $p\in K$ and $\lambda\in(0,1]$. This distribution $\mathcal{R}F\in\mathcal{D}^\prime(M)$ is non-unique.
\end{theorem}
\begin{proof}
The proof of this theorem is similar in spirit to that of Theorem \ref{Thm: reconstruction} and it uses a localized version of the same result of \cite{CZ20}. As a consequence, we only sketch  the proof. Moreover, we only discuss the case $\gamma<0$, the proof of the case $\gamma=0$ being analogous. As a consequence of \cite[Thm. 4.4]{CZ20}, for any $(U,\phi)\in\mathcal{A}$ there exists a distribution $(\mathcal{R}F)^{\phi(U)}\in\mathcal{D}^\prime(\phi(U))$ such that, for any compact set $K\subset U$ and for any $h\in\mathcal{D}(\phi(U))$
\begin{align}
|((\mathcal{R}F)^{\phi(U)}-\phi_*(F_p))(h_{\phi(p)}^\lambda)|\lesssim\lambda^\gamma\,,
\end{align}
uniformly for $p\in K$ and for $\lambda\in(0,1]$. 
Notice that, since $\gamma<0$, this distribution is non-unique. 
Nonetheless, we can choose for any $(U_j,\phi_j)\in\mathcal{A}$ a reconstructed local distribution $(\mathcal{R}F)^{\phi_j(U_j)}\in\mathcal{D}^\prime(\phi_j(U_j))$. 
We can now introduce a partition of unity $\{\rho_j\}_{j\in J}$ subordinated to the covering $\{U_j\}_{j\in J}$ of the manifold $M$. 
In this way, similarly to \cite[Sect. 11]{CZ20}, one can construct a global reconstructed distribution $\mathcal{R}_{\mathcal{A}, \rho}\in\mathcal{D}^\prime(M)$ satisfying the bound of Equation \eqref{Eq: reconstruction for gamma le0}. 
We underline that this distribution is non-unique. Indeed, it depends on the the choice of the local reconstructed distributions $(\mathcal{R}F)^{\phi_j(U_j)}\in\mathcal{D}^\prime(\phi_j(U_j))$, on the atlas $\mathcal{A}$ and on the partition of unity $\{\rho_j\}_{j\in J}$.
 The dependence on the partition of unity is a consequence of the lack of the overlapping condition, \textit{cf}. Equation \eqref{Eq: overlapping condition} \cite[Thm. 1.4.3]{FJ99}.
\end{proof}

\appendix
\section{Distributions on Smooth Manifolds}\label{Appendix: distribution on manifold}
In this appendix we shall recall some basic notions and results regarding distribution theory on smooth manifolds, in order to keep the paper self-contained. 
In particular, adopting \cite{Hor03} as a reference, we shall recall the definition of distribution on a smooth manifold $M$ and we recall also a very useful characterization of this concept, which we use extensively in the main body of the paper, in particular in the proof of the reconstruction theorem. 
\begin{definition}
Let $M$ be a $d$-dimensional smooth manifold. For any local chart $(U,\phi)$ on $M$, let  $t_{\phi(U)}\in\mathcal{D}^\prime(\phi(U))$ be a distribution satisfying the overlapping condition
\begin{align}\label{Eq: generic overlapping condition}
t_{\phi^\prime(U^\prime)}=(\phi\circ\phi^{\prime-1})^*t_{\phi(U)}\,,\qquad\mathrm{on}\quad\phi^\prime(U^\prime\cap U)\,,
\end{align}
we call the family $t_{\phi(U)}$ a distribution $t$ on the manifold $M$, denoting $t\in\mathcal{D}^\prime(M)$.
\end{definition}

The next theorem is very useful since it allows to verify the overlapping condition, Equation \eqref{Eq: overlapping condition}, only on one atlas in order to construct a global distribution in $\mathcal{D}^\prime(M)$ instead of considering all possible local charts over $M$.
\begin{theorem}\cite[Theor. 6.3.4]{Hor03}\label{Trm: Hormander}
Let $M$ be a smooth $d$-dimensional manifold and let $\mathcal{A}=\{(U_j,\phi_j)\}_{j\in J}$ be an atlas for $M$. Assume moreover that for any local chart $(U,\phi)\in\mathcal{A}$ there exists a distribution $t_{\phi(U)}\in\mathcal{D}^\prime(\phi(U))$ such that the overlapping condition,
\begin{align*}
t_{\phi^\prime(U^\prime)}=(\phi\circ\phi^{\prime-1})^*t_{\phi(U)}\,,\qquad\mathrm{on}\quad\phi^\prime(U^\prime\cap U)\,,
\end{align*}
holds true for any pair of local charts $(U,\phi),(U^\prime,\phi^\prime)\in\mathcal{A}$. Then there exists one and only one distribution $t\in\mathcal{D}^\prime(M)$ such that $\phi_*t=t_{\phi(U)}$ for any $(U,\phi)\in\mathcal{A}$.
\end{theorem}
Now we prove a standard result of distribution theory which we use in the main body of the paper. 
\begin{lemma}\label{Lemma: partial evaluation}
Let $U\subset\mathbb{R}^d$ be an open set and let $T\in\mathcal{D}^\prime(U)$ be a distribution such that, for any $K\subset U$ compact and any $h\in\mathcal{D}(K)$ such that $\int dx\,h(x)=1$, $|T(h_x^\lambda)|\to0$ for $\lambda\to0^+$, uniformly for $x\in K$. Then, for any $f\in\mathcal{D}(U)$, $T(f)=0$.
\end{lemma}
\begin{proof}
Let $K\subset U$ be a compact set. On account of the hypotheses, for any test-functions $u,h\in\mathcal{D}(K)$ such that $\int dx\,h(x)=1$, $u\ast h^\lambda\to u$ in $\mathcal{D}$ for $\lambda\to0^+$, where $\ast$ denotes the convolution. 
It follows, by sequential continuity of $T$, that $T(u\ast h^\lambda)\to T(u)$ for $\lambda\to0^+$. 
Furthermore,
\begin{align*}
|T(u\ast h^\lambda)|=\bigg|\int_{\mathbb{R}^d}dx\,u(x)T(h^\lambda_x)\bigg|\lesssim\|u\|_\infty\sup_{x\in K}|T(h_x^\lambda)|\,.
\end{align*}
Finally, on account of the hypothesis, we have $\sup_{x\in K}|T(h_x^\lambda)|\to0$ for $\lambda\to0^+$, implying $T(u)=0$. Since this argument holds true for any $K\subset U$ compact, this proves the thesis.
\end{proof}

\section{Coherence on an Open Set}\label{Appendix: enhanced coherence on an open set}
In this appendix we discuss the notion both of coherent germ of distributions on an open set $U\subset\mathbb{R}^d$ and of \emph{enhanced coherence} on an open set.
This ``local discussion'' will be useful to prove enhanced coherence on a smooth manifold.

We start by introducing the notion of coherence on an open set $U\subset\mathbb{R}^d$.
\begin{definition}
\label{Def:coherence on an open set}
Let $U\subset \mathbb{R}^d$ and let $\gamma \in \mathbb{R}$. We say that a germ of distributions $F=\{F_x\}_{x\in U}$, with $F_x\in\mathcal{D}^\prime(U)$ for any $x\in U$, is $\gamma$-coherent on $U$ if there exists a test function $\varphi \in \mathcal{D}(U)$ with $\int_Udx\, \varphi(x) \neq 0$ such that, for any compact set $K \subset U$, there exists $\alpha_K\leq \min\{\gamma,0\}$ for which
\begin{equation}
\label{eq:coherence}
\lvert (F_z - F_y)(\varphi^\varepsilon_y) \rvert \lesssim \varepsilon^{\alpha_K}(\lvert z- y \rvert + \varepsilon)^{\gamma-\alpha_K}\,,
\end{equation}
uniformly for $z,y \in K$ and for $\varepsilon \in \big(0,\frac{D_K}{4}\big]$, where $D_K:=\operatorname{dist}(\partial U,K)$. Here $\partial U$ denotes the boundary of $U$.
\end{definition}
\begin{remark}
Since $\partial U$ is a closed subset and $K$ is a compact subset with $\partial U \cap K = \emptyset$, then $D_K:=\operatorname{dist}(K,\partial U)>0$.
\end{remark}
\begin{remark}
In the previous definition, with respect to the case of a smooth manifold, \textit{cf}. Definition \ref{Def: coherence}, we exploited the argument of Remark \ref{Rmk: who cares about 1} for the supremum among the possible values taken by the scaling parameter $\varepsilon$, \emph{i.e.}, $\varepsilon \in \big(0,\frac{D_K}{4}\big]$. In particular, this choice of the supremum is convenient for the following discussion.
\end{remark}

\begin{remark}
Henceforth, we shall use the following notation. Given a compact set $H\subset\mathbb{R}^d$ and $\varepsilon>0$, we denote with $\bar{H}_\varepsilon$ the $\varepsilon$-enlargement of $H$, which is the set $\bar{H}_\varepsilon\vcentcolon=\{z\in\mathbb{R}^d\,:\,|z-x|\leq\varepsilon\,\mathrm{for}\,\mathrm{some}\,x\in H\}$. Notice that $\bar{H}_\varepsilon$ is compact.
\end{remark}
The idea at the base of \emph{enhanced coherence} is that of removing from the notion of coherence the dependence on the test-function $\varphi$. 
This can be achieved working in the same spirit of \cite{CZ20}, \emph{i.e.}, extending the class of test-functions paying the prize of suitably modifying some coherence parameters. 
As a premise, we state the following proposition.
\begin{proposition}
\label{prop:enhanced}
Let $U\subset \mathbb{R}^d$ be an open subset and let $T \in \mathcal{D}^\prime(U)$ be a distribution with the following property.
There is a compact subset $K\subset U$ and a test function $\varphi \in \mathcal{D}(U)$ with $\int_Udx\, \varphi(x) \neq 0$ such that, for all $x \in \overline{K}_{\frac{D}{2}}$ and for any $\varepsilon \in \{2^{-k}\}_{k\in\mathbb{N}}$
\begin{equation}
\label{eq:condition_for_uniformity}
\lvert T(\varphi^\varepsilon_x) \rvert \leq \varepsilon^\alpha f(\varepsilon,x)	\,.
\end{equation}
Here $\alpha \leq 0$, $f\colon \big(0,\frac{D}{4}\big]\times\overline{K}_{\frac{D}{2}}\to[0,\infty)$ is an arbitrary function, where $D:=\operatorname{dist(K,\partial U)}$ and $\partial U$ denotes the boundary of $U$ while $\overline{K}_{\frac{D}{2}}$ is the $\frac{D}{2}$-enlargement of $K$.
Then, for any integer $r > -\alpha$, for any $x\in K$ and for any $\psi \in \mathcal{D}(B(0,1))$,
\begin{equation}
\label{eq:uniformity}
\forall \lambda \in \biggl(0,\frac{D}{4}\biggr]\,, \qquad \lvert T(\psi^\lambda_x) \rvert \leq \mathfrak{b}_{\varphi,\alpha,r,d}\|\psi\|_{C^r}\lambda^\alpha \overline{f}(\lambda,x)\,,
\end{equation}
where $\mathfrak{b}_{\varphi,\alpha,r,d}$ is a constant while $\overline{f}\colon \big(0,\frac{D}{4}\big] \times K\to [0,\infty)$ is defined as
\begin{equation}
\overline{f}(\lambda,x):=\sup_{\lambda^\prime \in (0,\lambda], x^\prime \in B(x,2\lambda)} f(\lambda^\prime,x^\prime)\,.
\end{equation}
\end{proposition}

\begin{proof}
This proof follows slavishly  that of \cite[Prop. 12.6]{CZ20}. Hence, we do not report it. We highlight only the main difference. This lies in the fact that, being this result the localization on an open set $U$ of \cite[Prop. 12.6]{CZ20}, when we consider the enlargement of the compact set $K$ we need to make sure that this is still contained in $U$. This justifies the introduction of $D:=\operatorname{dist(K,\partial U)}$. Indeed, the $\frac{D}{2}$-enlargement of $K$ is, per construction, a compact set contained in $U$.
\end{proof}

\noindent As a consequence, we have \emph{enhanced coherence}.
\begin{proposition}[Enhanced coherence]
\label{Prop:enhanced_coherence on an open set}
Let $U \subset \mathbb{R}^d$ and let $F=\{F_x\}_{x\in U}$ be a $\gamma$-coherent germ of distributions on $U$, \textit{i.e.} let Equation \eqref{eq:coherence} hold true for some $\varphi\in \mathcal{D}(U)$ and some family $\bm{\alpha}=(\alpha_K)$. Define
\begin{equation}
\label{eq:new_alpha}
\tilde{\bm{\alpha}}=(\tilde{\alpha}_K)\qquad \text{where} \quad \tilde{\alpha}_K = \alpha_{\overline{K}_{\frac{D_K}{2}}} \quad\text{and}\quad D_K=\operatorname{dist}(K,\partial U)\,.
\end{equation}
Then, for any compact set $K\subset U$ and any $r > -\tilde{\alpha}_K$,
\begin{equation}
\label{eq:enhanced_coherence}
\lvert(F_z-F_y)(\psi^\varepsilon_y)\rvert \lesssim \|\psi\|_{C^r} \varepsilon^{\tilde{\alpha}_K}(\lvert z-y \rvert + \varepsilon)^{\gamma-\tilde{\alpha}_K}\,,
\end{equation}
uniformly for $z,y \in K$, $\varepsilon \in \big(0,\frac{D_K}{4}\big]$ and $\psi\in\mathcal{D}(B(0,1))$, where $B(0,1)\subset\mathbb{R}^d$ denotes the unitary ball centred at the origin. 
\end{proposition}

\begin{proof}
The proof follows the same lines of \cite[Prop. 13.1]{CZ20}. As above, we do not report it. We highlight only the main difference. This lies in the fact that when we consider the enlargement of the compact set $K$ we need to make sure that this is still contained $U$. This is guaranteed by the definition of $D_K$. This, together with the notion of coherence on $U$, guarantees the existence of the coherence parameters $\alpha_{\overline{K}_{\frac{D_K}{2}}}$ associated with the $\frac{D_K}{2}$-enlargement of $K$.
\end{proof}

Proposition \ref{Prop:enhanced_coherence on an open set} shows that \emph{coherence} implies \emph{enhanced coherence}. The inverse implication holds true trivially. As a consequence, we have the following equivalent definition of coherence on an open subset $U\subset\mathbb{R}^d$.
\begin{definition}
\label{def:coherence-final}
Let $U\subset \mathbb{R}^d$ be an open subset and let $\gamma \in \mathbb{R}$. We say that a germ of distributions $F=\{F_x\}_{x\in U}$ is $\gamma$-coherent on $U$ if for any compact set $K \subset U$ there exists a real number $\alpha_K\leq \min\{\gamma,0\}$ such that, for any $r>-\alpha_K$,
\begin{equation}
\label{eq:coherence-final}
\lvert (F_z - F_y)(\psi^\varepsilon_y) \rvert \lesssim\|\psi\|_{C^r} \varepsilon^{\alpha_K}(\lvert z- y \rvert + \varepsilon)^{\gamma-\alpha_K}
\end{equation}
uniformly for $z,y \in K$ and for $\varepsilon \in \big(0,\frac{D_K}{4}\big]$, where $D_K:=\operatorname{dist}(\partial U,K)$ for any $\psi\in\mathcal{D}(B(0,1))$.
\end{definition}
The notion of coherence on an open set is stable with respect to restrictions of the open set itself.
\begin{proposition}
\label{Prop:restriction on an open set}
Let $U \subset \mathbb{R}^d$ be an open set and let $V\subset U$ be an open subset. If a germ of distributions $F=\{F_x\}_{x\in U}$ is $\gamma$-coherent on $U$, then it is $\gamma$-coherent also on $V$.
\end{proposition}
\begin{proof}
Let  $K\subset V$ be a compact set and define $D_K^U\vcentcolon=\operatorname{dist(\partial U,K)}$ while $D_K^V\vcentcolon=\operatorname{dist}(\partial V, K)$. Being  $K \subset U$, there exists $\alpha_K^U\leq \min\{\gamma,0\}$ such that, for any $r > -\alpha_K^U$,
\begin{equation}
\label{eq:coherence_V}
\lvert(F_z-F_y)(\psi^\varepsilon_y)\rvert \lesssim \|\psi\|_{C^r} \varepsilon^{\alpha_K^U}(\lvert z-y \rvert + \varepsilon)^{\gamma-\alpha_K^U} 
\end{equation}
uniformly for $z,y \in K$, $\varepsilon \in \big(0,\frac{D^U_K}{4}\big]$ and $\psi\in\mathcal{D}(B(0,1))$. Since $D_K^V \leq D_K^U$, the bound in Equation \eqref{eq:coherence_V} holds true uniformly for $\varepsilon \in \big(0,\frac{D^V_K}{4}\big]$, \emph{i.e.}, the germ of distributions $F$ is $\gamma$-coherent on $V$ with parameters $(\alpha^U_K)$.
\end{proof}


\begin{thebibliography}{}
		\bibitem[BGP07]{BGP07}
		C. B\"ar, N. Ginoux, F. Pf\"affle, 
		\textit{``Wave Equations on Lorentzian Manifolds and Quantization''},
		1 st edn, Eur. Math. Soc., Z\"urich, 2007
		
		\bibitem[BFDY]{BFDY15}
	R.~Brunetti, C.~Dappiaggi, K.~Fredenhagen, Y.~Yngvason editors,
	{``Advances in Algebraic Quantum Field Theory''}, Mathematical Physics Studies 
	(2015) Springer, 455p.	
	
		\bibitem[BF09]{book}
	R.~Brunetti and K.~Fredenhagen,
	\textit{``Quantum field theory on curved backgrounds''},
	in {\em Quantum Field Theory on Curved Spacetimes}, Lecture Notes in Phys., Vol. {\bf 786} Springer, (2009), pp. 129--155.
	
		\bibitem[CZ20]{CZ20}
		F. Caravenna, L. Zambotti, 
		\textit{``Hairer's reconstruction theorem without regularity structures''},
		To appear in EMS Surveys in Math. Sci., 
		arXiv:2005.09287 (2020)
		
		\bibitem[DDK19]{DDK19}
		A. Dahlqvist, J. Diehl, B.K. Driever,
		\textit{``The Parabolic Anderson Model on Riemannian Surfaces''},
		Probability Theory and Related Fields (2019) 174: 349-444, 
		
		\bibitem[DDRZ20]{DDRZ20}
		C.~Dappiaggi, N.~Drago, P.~Rinaldi and L.~Zambotti,
		\textit{``A microlocal approach to renormalization in stochastic PDEs''},
		arXiv:2009.07640 (2020)
		
		\bibitem[DDR20]{DDR20}
	C.~Dappiaggi, N.~Drago and P.~Rinaldi,
	\textit{``The algebra of Wick polynomials of a scalar field on a Riemannian manifold''},
	Rev. Math. Phys. \textbf{32} (2020) no.08, 2050023,
	[arXiv:1903.01258 [math-ph]].
		
		\bibitem[FR16]{Fredenhagen:2014lda}
		K.~Fredenhagen and K.~Rejzner,
		\textit{``Quantum field theory on curved spacetimes: Axiomatic framework and 							examples''},
		J. Math. Phys. \textbf{57} (2016) no.3, 031101,
		[arXiv:1412.5125 [math-ph]].
		
		\bibitem[FJ99]{FJ99}
		F.G.~ Friedlander, M.~Joshi,
		\textit{``Introduction to the theory of distributions''}		,
		Cambridge University Press (1999)
		
		\bibitem[Hai14]{Hai14}
		M.~Hairer,
		\textit{``A theory of regularity structures''},
		Inv. Math. {\bf 198} (2014), 269,
		arXiv:1303.5113 [math.AP].
		
		\bibitem[H\"or03]{Hor03}
		L. H\"ormander, 
		\textit{The Analysis of Linear Partial Differential Operators I},
		Springer Berlin	(2003).
		
			\bibitem[PW81]{PW81}
	G.~Parisi and Y.~s.~Wu,
	\textit{``Perturbation Theory Without Gauge Fixing,''}
	Sci. Sin. \textbf{24} (1981), 483.
\end{thebibliography}
\end{document}